\documentclass[conference,onecolumn]{IEEEtran}

\IEEEoverridecommandlockouts

\usepackage{amsmath,amsfonts,amssymb, amsthm}
\usepackage{graphicx}
\usepackage{comment, color}

\graphicspath{ {./figures/} }
\newcommand{\myfigswidth}{.35\textwidth}

\newtheorem{theorem}{Theorem}
\newtheorem{lemma}{Lemma}

\newtheorem{definition}{Definition}

\newcommand{\floor}[1]{\left\lfloor{#1}\right\rfloor}
\newcommand{\ceil}[1]{\left\lceil{#1}\right\rceil}

\newcommand{\remove}[1]{}
\renewcommand{\epsilon}{\varepsilon}

\newcommand{\myshort}[1]{}
\newcommand{\myextended}[1]{#1}

\begin{document}
%
\title{Effect of Number of Users in Multi-level Coded Caching}

\author{\IEEEauthorblockN{Jad Hachem}
\IEEEauthorblockA{University of California, Los Angeles\\
Email: jadhachem@ucla.edu}
\and
\IEEEauthorblockN{Nikhil Karamchandani}
\IEEEauthorblockA{Indian Institute of Technology, Bombay\\
Email: nikhilk@ee.iitb.ac.in}
\and
\IEEEauthorblockN{Suhas Diggavi}
\IEEEauthorblockA{University of California, Los Angeles\\
Email: suhas@ee.ucla.edu}
\thanks{This work was supported in part by NSF grant \#1423271 and a gift by Qualcomm Inc.}
}


\maketitle

\begin{abstract}
It has been recently established that joint design of content delivery and storage (coded caching) can significantly improve performance over conventional caching.
This has also been extended to the case when content has non-uniform popularity through several models.
In this paper we focus on a multi-level popularity model, where content is divided into levels based on popularity.
We consider two extreme cases of user distribution across caches for the multi-level popularity model: a single user per cache (single-user setup) versus a large number of users per cache (multi-user setup).
When the capacity approximation is universal (independent of number of popularity levels as well as number of users, files and caches), we demonstrate a dichotomy in the order-optimal strategies for these two extreme cases.
In the multi-user case, sharing memory among the levels is order-optimal, whereas for the single-user case clustering popularity levels and allocating all the memory to them is the order-optimal scheme.
In proving these results, we develop new information-theoretic lower bounds for the problem.

\end{abstract}

\section{Introduction}
\label{sec:intro}
Wireless traffic has been dominated by broadband content access (driven by video applications) and has strained current wireless network capacity.
While there have been tremendous improvements in wireless data rates over successive generations of wireless systems, these gains alone are not projected to keep up with exponential rise in wireless data demand.
Pre-fetching and storing content in edge caches is one strategy that helps reduce network traffic \cite{FemtoCaching}.
Recently, it has been shown that joint design of storage and delivery (a.k.a.\ ``coded caching'') can significantly improve content delivery rate requirements \cite{maddah-ali2012}.
This was enabled by content placement that creates (network-coded) multicast opportunities among users with access to different storage units, even when they have different (and \emph{a priori} unknown) requests.
Coded caching has been shown to be well-suited to next-generation (heterogeneous) wireless network architectures \cite{MLcodedcaching,HKDinfocom15}.

The setup studied in \cite{maddah-ali2012,maddah-ali2013} consisted of single-level content, \emph{i.e.}, every file in the system is uniformly demanded.
However, it is well understood that content demand is non-uniform in practice, with some files being more popular than others.
Motivated by this, \cite{niesen2013,Zcodedcaching,ZhangArbitrary,MLcodedcaching} considered such non-uniform content demand, following different models.

In \cite{niesen2013,Zcodedcaching,ZhangArbitrary}, the setup considered a single user per cache requesting a file independently and randomly according to some arbitrary probability distribution that represents content popularity.
These works studied the trade-off between the average rate and the cache memory.
A memory-sharing scheme was proposed in \cite{niesen2013}, and its achievable rate was characterized.
However, from our understanding, this scheme was not shown to be order-optimal in general.%
\footnote{We refer to an ``order-optimal'' result as one that is within a constant multiplicative factor from the information-theoretic optimum.}
In \cite{Zcodedcaching,ZhangArbitrary}, a different scheme was proposed, based on a clustering of the most popular levels.
It was shown to be order-optimal for Zipf-distributed content in \cite{Zcodedcaching}, and, more recently, for arbitrary distributions in \cite{ZhangArbitrary}.

By contrast, in \cite{MLcodedcaching}, a deterministic multi-level popularity model was introduced, where it is assumed that a large number of users connect to each cache.
Content is divided into discrete levels based on popularity, and, for each level, a fixed and \emph{a priori} known fraction of the users per cache request files from said level.
It is easy to see that, when the number of users per cache is large enough, this deterministic model will closely approximate an equivalent stochastic-demands model similar to \cite{niesen2013,Zcodedcaching,ZhangArbitrary}.
A worst-case rate-memory trade-off was studied, where the ``worst case'' is over all user demand tuples that obey the constraints set by the model, \emph{i.e.}, where there is a fixed number of users per popularity level per cache.
A memory-sharing strategy was shown to be order-optimal,%
\footnote{The approximation was, however, dependent on the number of popularity levels; an aspect that is strengthened in this paper.}
and a parametric characterization of how much memory to allocate to each popularity level was established.
It was shown that sometimes it is better to store some less popular content without completely storing the more popular content.

In this paper, we focus our attention on the multi-level popularity model, and we explore the role of the number of users per cache by contrasting two extreme cases: a single user per cache with an overall fixed profile of content request \emph{across} all caches (called the ``single-user'' setup in this paper) versus a large number of users per cache with a fixed content request profile for \emph{every} cache, as in \cite{MLcodedcaching} (the ``multi-user'' setup).
We ask for an order-optimality guarantee that is \emph{independent} of the number of popularity levels (in addition to the number of users, files, and caches); a stronger requirement than in \cite{MLcodedcaching}.
For such an approximation guarantee, we demonstrate that drastically different strategies must be used for each setup.
In the single-user case, we show that \emph{clustering} the most popular levels and giving them all the memory, leaving none for the rest, is order-optimal; a strategy similar to those proposed in \cite{Zcodedcaching,ZhangArbitrary}.
In contrast, the multi-user case requires a complete separation of the different levels and a division of the memory between them: the \emph{memory-sharing} scheme introduced in \cite{MLcodedcaching}.
In order to prove these results, we needed to develop new information-theoretic lower bounds for both problems; in particular, the multi-user setup necessitated new non-cut-set bounds that use sliding-window entropy inequalities \cite{Liu14}, providing much stronger approximation guarantees than the results in \cite{MLcodedcaching}.

We also discuss the dichotomy between the two setups.
We show that the memory-sharing strategy can perform poorly in the single-user setup, while clustering can be arbitrarily worse than memory-sharing in the multi-user case.
Furthermore, we briefly explore an open problem that is part of our ongoing research: that of solving a mixture of the two extreme setups.

The paper is organized as follows.
Section~\ref{sec:setup} formulates the problem, describing precisely the multi-user and single-user setups.
We establish some background in Section~\ref{sec:preliminaries}, which enables us to state the main results in Section~\ref{sec:results}.
Section~\ref{sec:comparison} gives an interpretation of the results and the dichotomy in the two setups and briefly explores the mixture of the setups.
\myshort{For lack of space, we skip the detailed proofs here and provide them instead in \cite{NumUsersExtended}.}
\myextended{Detailed proofs are given in the appendices.}

\section{Setup}
\label{sec:setup}

Consider a system where a group of users request files from a server, according to some popularity model.
All files are assumed to be of size $F$ bits.
Prior to any user requests, a \emph{placement phase} occurs in which information about these files is placed in the $K$ access point (AP) caches, of capacity $MF$ bits each.
Then, in the \emph{delivery phase}, users connect to the different caches, and each requests a file based on the popularity model; the more popular files are more likely to be requested.
The server then sends, through the base station (BS), a broadcast message of size $RF$ bits that all the users can hear.
The users combine the broadcast with the contents of their cache to recover the file that they have requested.
Clearly, there is a trade-off between the values of $M$ (the ``cache memory'') and $R$ (the ``broadcast rate'').
The larger the caches, the more information they can store, and hence the smaller the broadcast needed to serve the requests.

The popularity model that we consider is the \emph{multi-level} model.
The files are divided into $L$ \emph{popularity levels}, such that all files in a single level are equally popular.
The levels consist of $N_1,\ldots,N_L$ files.
When a large enough number of users are present in the system, we expect a crystallization of the user profile with respect to the popularity levels.
Specifically, we expect to know, to some degree, the number of users that are making requests from each level.
For example, suppose there are two popularity levels such that a randomly chosen user is three times as likely to request a file from the first level as he is from the second.
Suppose there are $40$ users in the network.
Then, we would expect that about $30$ of them will request a file from the first level, and $10$ from the second.
The multi-level popularity model assumes that these numbers are \emph{fixed} and \emph{known} for every level.
The large number of users allows this deterministic model to closely resemble the stochastic-demands model.
This enables us to analyze the rate-memory trade-off in the worst case, among all cases where user demands respect the constraints imposed by this model.

A pair $(R,M)$ is said to be \emph{achievable} if there exists a placement-and-delivery strategy that uses caches of memory $M$ and transmits, for \emph{any} possible combination of user requests, a broadcast message of rate at most $R$ that satisfies all said requests.
Our goal is to find all such achievable pairs.
In particular, we wish to find the optimal rate-memory trade-off:
\[
R^\ast(M) = \inf\left\{ R : \text{$(R,M)$ is achievable} \right\},
\]
where the minimization is done over \emph{all} possible strategies.

In this paper, we study how the number of users in the system affects its overall behavior.
In terms of the setup, the difference between a large and small number of users is a difference in the distribution of the popularity levels across the users.
For illustration, consider again the example above with two levels and a user ratio of $3$ to $1$, and suppose that there are $20$ caches in the network.
If every cache had just one user connecting to it, then we would expect that about $15$ of all the $20$ users will request files from the first level, and $5$ would request files from the second level, giving a situation similar to \figurename~\ref{fig:setup-G}.
However, we cannot predict which users will be at which cache.
On the other hand, suppose there were $20$ users at every cache (for a total of $400$ users).
Then, we would predict that, \emph{at every cache}, approximately $15$ users will request files from the first popularity level and $5$ will request files from the second.
Hence, the request profiles would be roughly similar for all caches, as illustrated in \figurename~\ref{fig:setup-F}.
We stress again that the multi-level model is deterministic, and that these predictions are in fact assumed as fixed and known realizations.

\subsection{Multi-user setup}

\begin{figure}
\centering
\includegraphics[width=\myfigswidth]{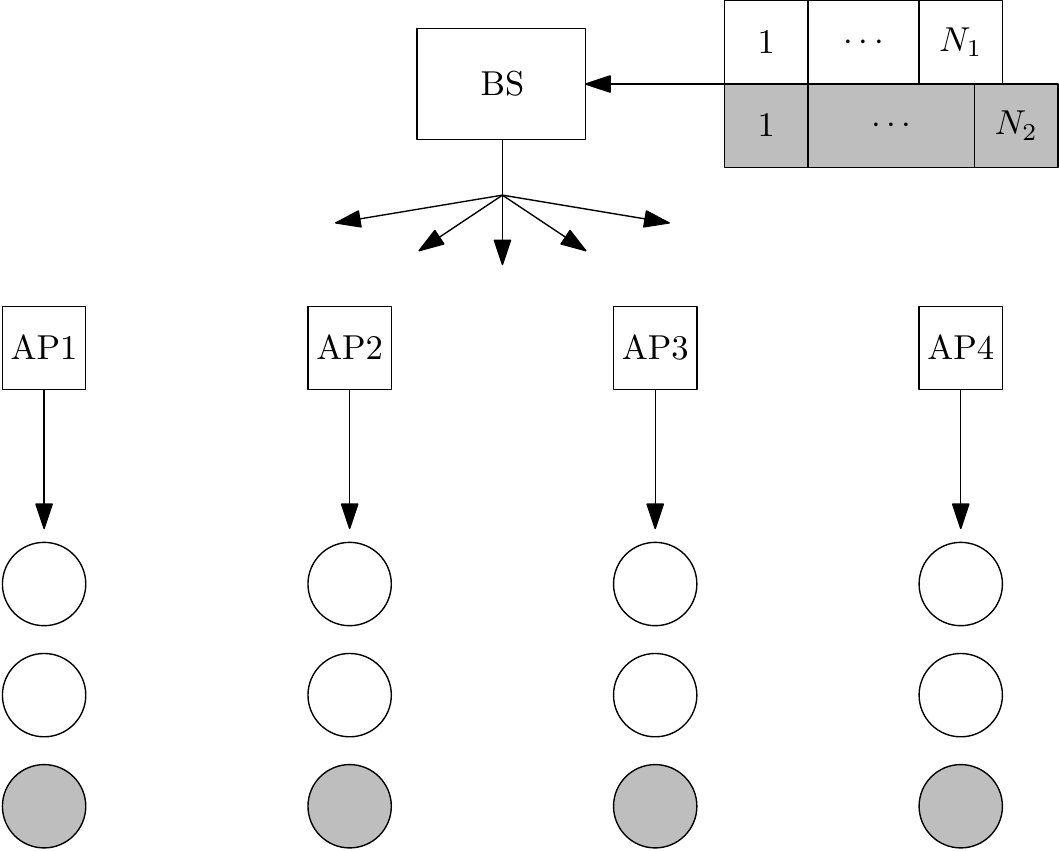}
\caption{Multi-user setup with $K=4$ caches, and $L=2$ levels with $(U_1,U_2)=(2,1)$ users per cache.}
\label{fig:setup-F}
\end{figure}

Consider the setup shown in \figurename~\ref{fig:setup-F}.
For every level $i$, we have exactly $U_i$ users connecting to every cache and requesting a file from $i$.
Notice that every level is represented at every cache.
This setup is identical to the one studied in \cite{MLcodedcaching}.

We assume the following two regularity conditions.
First, for every popularity level $i$, there are more files than users:
\begin{equation}
\label{eq:mu-reg1}
\forall i,\quad N_i\ge KU_i.
\end{equation}
This can be seen, for example, in video applications such as Netflix, where ``files'' would be video segments of a few seconds to a few minutes.

Second, we assume that no two levels have similar popularities.
The popularity of a level can be written as the number of users per file of the level.
Hence, if $i$ is a more popular level than $j$, the regularity condition states:
\begin{equation}
\label{eq:mu-reg2}
\frac{U_i/N_i}{U_j/N_j} \ge \frac{1}{\beta^2},
\end{equation}
where $\beta=1/80$.
The reasoning behind this condition is that, if it did not hold for some levels $i$ and $j$, then we can think of them as essentially one level with $N_i+N_j$ files and $U_i+U_j$ users per cache.
The resulting popularity $\frac{U_i+U_j}{N_i+N_j}$ would be close to both $U_i/N_i$ and $U_j/N_j$.

\subsection{Single-user setup}

\begin{figure}
\centering
\includegraphics[width=\myfigswidth]{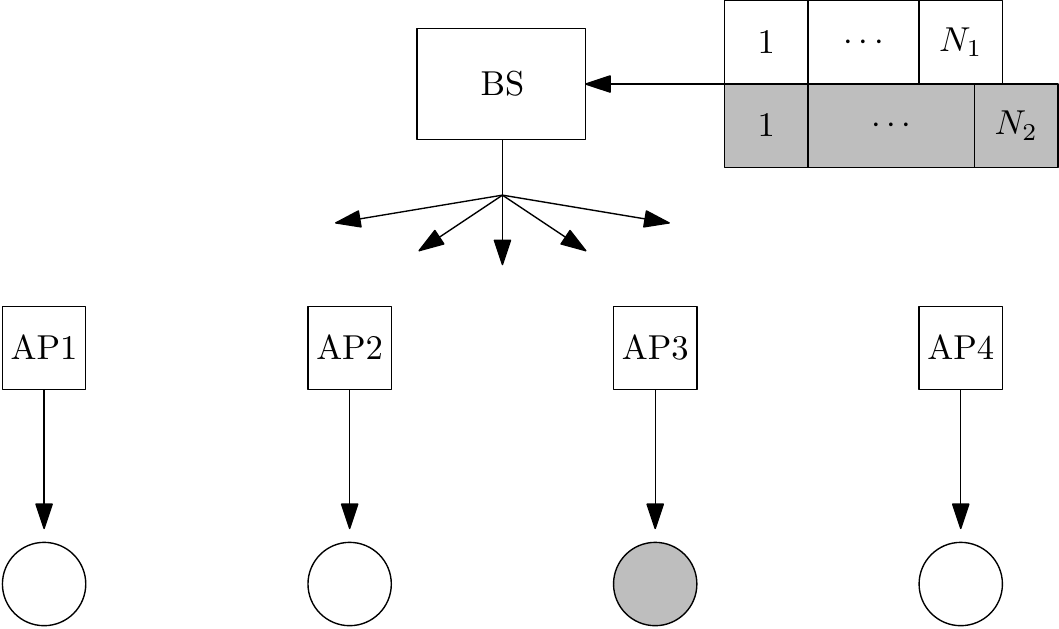}
\caption{Single-user setup with $K=4$ caches, and $L=2$ levels with $(K_1,K_2)=(3,1)$ users.}
\label{fig:setup-G}
\end{figure}

Consider now the setup in \figurename~\ref{fig:setup-G}, depicting the other extreme.
We have only one user connecting to every cache, for a total of $K$ users.
The only information known \emph{a priori} is that, for each level $i$, exactly $K_i$ out of the $K$ users will request a file from $i$.
However, we do not know which users these will be.
The achievability scheme that we will design must allow for any possible arrangement of the users.

In this setup, we also assume that we always have more files than users.
In particular, for any level $i$:
\begin{equation}
\label{eq:su-reg}
N_i \ge K_i.
\end{equation}

\section{Preliminaries}
\label{sec:preliminaries}
Coded caching was introduced in \cite{maddah-ali2012}, which established its order-optimality for a single level of popularity and a single user per cache.
The extension to multiple users per cache was looked at in \cite{maddah-ali2012,maddah-ali2013,JiMultiuser} and is a special case of \cite{Hcodedcaching}; it can be formally stated as follows.

\begin{lemma}[Adapted from {\cite[Lemma~3]{MLcodedcachingExtended}}]
\label{lemma:single-level}
For a single-level caching system with $K$ caches, $U$ users at every cache, $N$ uniformly-distributed files, and a cache memory of $M$, the following rate is achievable:
\[
R^\text{SL}(M,K,N,U) = U \cdot \min\left\{ \frac{N}{M} , K \right\} \cdot \left( 1-\frac{M}{N} \right).
\]
Furthermore, this rate is within a constant of the optimum.
\end{lemma}

In \cite{MLcodedcaching}, we developed a \emph{memory-sharing} scheme for content with multi-level popularity model.
The strategy consisted of dividing the cache memory between all the $L$ levels, and then treating each level as a separate caching sub-system, with the reduced memory.
In other words, we give level $i$ a memory $\alpha_iM$, where
$\alpha_i\in[0,1]$ and $\sum_i\alpha_i=1$, and then apply a single-level placement-and-delivery strategy for this level on this $\alpha_iM$ memory, separately from the other levels.
The total rate for this scheme is:
\begin{equation}
R^\text{MU}\left(M,K,\{N_i,U_i\}_i\right)
= \sum_{i=1}^L R^\text{SL}(\alpha_iM,K,N_i,U_i).
\end{equation}

An optimization of the memory-sharing parameters $\{\alpha_i\}_i$ was studied in \cite{MLcodedcaching}, and an allocation which was demonstrated to be order-optimal was established.
At a high-level this was done by partitioning the popularity levels into three sets: $H$; $I$; and $J$.
The levels in $H$ have such a small popularity that they will get no cache memory.
On the opposite end of the spectrum, the levels in $J$ will get enough cache memory to completely store all their files in every cache.
Finally, the rest of the levels, in the set $I$, will share the remaining memory among themselves, obtaining some non-zero amount but not enough to store all of their files.
An order-optimal $(H,I,J)$ partition and corresponding memory assignments were established in \cite{MLcodedcaching}.

\begin{theorem}[Adapted from {\cite[equation (2)]{MLcodedcaching}}]
\label{thm:multi-user-achievability}
Given a multi-user caching setup, with $K$ caches, $L$ levels, and, for each level $i$, $N_i$ files and $U_i$ users per cache, and a cache memory of $M$, the following rate%
\footnote{This expression of the rate is a slight approximation that we use here for simplicity as it is more intuitive.}  is achievable:
\[
R^\text{MU}(M) \approx
\sum_{h\in H} KU_h
+ \frac{ \left( \sum_{i\in I} \sqrt{N_iU_i} \right)^2 }{ M - \sum_{j\in J}N_j }
- \sum_{i\in I}U_i,
\]
where $(H,I,J)$ is a particular type of partition of the set of levels called an $M$-feasible partition.%
\footnote{See \cite{MLcodedcaching} for more details.}
\end{theorem}

Intuitively, since a level $h\in H$ receives no cache memory, all requests from its $KU_h$ users must be handled directly from the broadcast, which requires the transmission of $KU_h$ complete files.
The users in set $J$ require no transmission as the files are completely stored in all the caches; however, set $J$ does affect the rate through the memory available for levels in $I$.
This is apparent in the expression $M-\sum_{j\in J}N_j$.
Finally, the levels in $I$, having received some memory, result in a rate that is inversely proportional to the effective memory and that depends on the level-specific parameters $N_i$ and $U_i$.

\section{Main results}
\label{sec:results}
\subsection{Multi-user setup}

The memory-sharing scheme was proved to be within an approximation factor of $\Theta(L^3)$ of the information-theoretic optimum, where
$L$ is the number of levels \cite{MLcodedcaching}.
In the following theorem, we tighten this gap by developing new, non-cut-set lower bounds, which use sliding-window entropy inequalities introduced in \cite{Liu14} and take into account the contribution of all the levels to the rate.
With these bounds, we are able to completely remove the dependence on the number of levels~$L$.

\begin{theorem}\label{thm:multi-user-gap}
For all valid values of the problem parameters $K$, $L$,
$\{N_i,U_i\}_i$, and $M$, we have:
\[
\frac{R(M)}{R^\ast(M)} \le 192,
\]
where $R(M)$ is the rate achieved by memory-sharing, and $R^\ast(M)$
is the optimal rate over all strategies.
\end{theorem}

\subsection{Single-user setup}

In the single-user setup, the scheme that we propose is quite different.
Instead of separating the levels, we \emph{cluster} a subset of them into a super-level that will be treated as essentially one level.
Specifically, we partition the levels into two subsets: $H'$ and $I'$.
The set $I'$ will be clustered as one super-level and given the entire memory $M$, while $H'$ will receive no memory.

To understand how to choose $H'$ and $I'$, consider the following rough analysis.
Suppose that all levels except one (let us call it $j$) have been split into $H'$ and $I'$.
Then, ignoring level $j$ for the moment, the rate, using Lemma~\ref{lemma:single-level}, would be:
\begin{IEEEeqnarray*}{rCl}
R
&=& {\textstyle R^\text{SL}(0,\sum_{h\in H'}K_h,\sum_{h\in H'}N_h,1)}\\
&& \textstyle {}+ R^\text{SL}(M,\sum_{i\in I'}K_i,\sum_{i\in I'}N_i,1)\\
&\approx& \sum_{h\in H'} K_h
+ \frac{\sum_{i\in I'}N_i}{M}.
\end{IEEEeqnarray*}
If we were to add level $j$ to $H'$, that would result in the addition of a $K_j$ term, since all $K_j$ requests would be completely served by the broadcast.
On the other hand, if it is added to $I'$, then we would get an additional $N_j/M$ term, since the total number of files in $I'$ would increase by $N_j$.
Clearly, it is beneficial to choose the smaller of the two quantities.

Though the above analysis is rough, its main idea still holds.
In general, we choose the partition $(H',I')$ as follows:
\begin{equation}
\label{eq:G-partition}
H' = \left\{ h\in\{1,\ldots,L\} : M < \frac{N_h}{K_h} \right\};
\quad
I' = (H')^c.
\end{equation}
Then, by giving all of the memory to $I'$, we can apply a single-level caching-and-delivery scheme to obtain the rate in the following theorem.

\begin{theorem}\label{thm:single-user-achievability}
Consider the multi-level, single-user setup with $L$ levels, $N_i$ files and $K_i$ users for each level $i$, and cache memory $M$.
Then, the following rate is achievable:
\[
R^\text{SU}\left(M,\{N_i,K_i\}_i\right)
=
\sum_{h\in H'} K_h
+
\max\left\{ \frac{\sum_{i\in I'}N_i}{M} - 1 \,,\, 0 \right\},
\]
where $H'$ and $I'$ are as in \eqref{eq:G-partition}.
\end{theorem}

The next theorem shows that this scheme is order-optimal for the single-user setup.

\begin{theorem}\label{thm:single-user-gap}
In the single-user setup, if $R(M)$ denotes the rate achieved by the clustering scheme and $R^\ast(M)$ denotes the optimal rate, then, for all values of the problem parameters $L$, $\{N_i,K_i\}_i$ and $M$:
\[
\frac{R(M)}{R^\ast(M)} \le 72.
\]
\end{theorem}

This result can be proved using cut-set bounds.


\section{Comparison}
\label{sec:comparison}

In this section, we first compare the memory-sharing and the clustering strategies, and we explore the dichotomy among the two setups that is emphasized by the difference between strategies.
We will then discuss why such a dichotomy exists, and explain the need for different lower bounds for each setup.
Finally, we explore a new problem that combines both setups by including both multi-user and single-user levels.

\subsection{Comparing the two caching-and-delivery strategies}

We have previously argued that memory-sharing is the best scheme to use in the multi-user case, while clustering is the near-optimal strategy in the single-user case.
However, why could one (or both) of these schemes not be good enough for both situations?
We will show, in this section, how that is not the case: memory-sharing can give an $L$-dependent gap between its rate and the optimum when used in the single-level setup; meanwhile, the rate achieved by clustering in the multi-user case can be arbitrarily far from the optimal rate.
We give examples of these two cases.

Consider a multi-user setup with two levels such that $(N_1,N_2)=(2^{5r},2^{8r})$ and $(U_1,U_2)=(2^{4r},2^r)$ for some $r>0$.
Suppose that there is enough memory so that both levels are to be partially stored in the caches.
With the memory-sharing scheme, that would give a rate of approximately:
\begin{IEEEeqnarray*}{rCl}
R
&\approx& \frac{\left( \sqrt{N_1U_1}+\sqrt{N_2U_2} \right)^2}{M}
= \frac1M \cdot \Theta\left( 2^{9r} \right).
\end{IEEEeqnarray*}
On the other hand, if we had clustered the two levels into one, then this super-level would have $(N_1+N_2)$ files and $(U_1+U_2)$ users per cache, resulting in the following rate:
\begin{IEEEeqnarray*}{rCl}
R
&\approx& \frac{(N_1+N_2)(U_1+U_2)}{M}
= \frac1M \cdot \Theta\left( 2^{12r} \right).
\end{IEEEeqnarray*}
Clearly, the latter rate is about $2^{3r}$ times as large as the former, a ratio that can get arbitrarily large as $r$ increases.

This difference is most pronounced when the popularities of the two levels become significantly different (in the above example, the popularity ratio was $\frac{U_1/N_1}{U_2/N_2}=2^{6r}$).
Intuitively, if the two levels had similar popularities, then memory-sharing gives them similar amounts of memory, effectively merging them.
However, if their popularities were very different, then they should be given highly unequal portions of the memory.

Consider now the single-user case with $L$ levels, and suppose again that the memory is such that all levels will be partially stored.
Let us assume that $N_1=\cdots=N_L$.
Using the clustering scheme, we get the following approximate rate:
\[
R \approx \frac{N_1+\cdots+N_L}{M}
= \frac{LN_1}{M}.
\]
However, with memory-sharing, we would get:
\[
R \approx \frac{\left( \sqrt{N_1}+\cdots+\sqrt{N_L} \right)^2}{M}
= \frac{L^2N_1}{M},
\]
which is larger by a factor of $L$.
Essentially, we are sending $L$ broadcasts, one per level, when we could send just one broadcast for all $L$ levels.

\subsection{Analysis of the dichotomy between the setups}

The dichotomy between the two extremes is striking.
They require drastically different strategies, and the strategy that is good for one setup is not so for the other.
This suggests a fundamental difference between the two setups.

To understand this difference, consider what happens when sending a coded broadcast message.
Each message targets a specific subset of users.
If, in this subset, there exist two users that are connected to the same cache, then these users have access to the exact same side information.
As a result, no coding can be done across these two users, and there is hence no use in including them in the same broadcast.

With that in mind, consider again \figurename~\ref{fig:setup-F} and \figurename~\ref{fig:setup-G}.
Notice how, in the multi-user setup, there are multiple rows of users, each of which consists of users from the \emph{same} popularity level.
Each such row is a complete set of users with no common caches: any additional users would have access to the same cache as some other user.
Thus, it is sufficient to consider them in a broadcast transmission that is separate from all other rows.
Since, as a result, no two levels will share the same broadcast message, it can only be beneficial to choose the best possible division of the memory, based on popularities.

In the single-user setup, however, there is only one row of users that contains all the users from all the levels.
It is hence possible to generate coding opportunities across levels.
Merging is thus a better option in this situation, and merging is most efficient when all levels receive equal memory per file.

\subsection{The difference in the lower bounds}

The reason different types of lower bounds are needed for the two setups is similar to the reason for the dichotomy in their respective caching-and-delivery strategies.
In the single-level setup studied in \cite{maddah-ali2012}, cut-set bounds were given to lower-bound the optimal rate.
Depending on the value of the rate, a certain number of caches were considered and used in the cut-set bounds.

When transitioning to the multi-level, multi-user scenario, we get a concatenation of broadcast messages, resulting in a sum of single-level rates.
Since these rates have potentially different values, each requires a cut-set bound that considers a different number of caches.
Thus we need lower bounds that consist of sums of cut-set bounds, each considering a different number of caches; a single cut-set bound is not enough.

However, in the single-user case, we are again faced with a single broadcast message to all users.
Thus, one cut-set bound is enough to give a lower bound on this rate.

\subsection{Mixing the setups}

So far, we have looked at the two extremes: either all levels were represented at all the caches, or none of them were.
A natural problem arises: that of studying intermediate cases.
The simplest form such intermediate cases can take is one where levels of both types are present.

Specifically, let there be two classes of popularity levels: $\mathcal{F}$ and $\mathcal{G}$.
The class $\mathcal{F}$ consists of levels $i$ that are represented by exactly $U_i$ users \emph{at every cache}.
In contrast, there is exactly one row of users that represents all the levels in the class $\mathcal{G}$: each level $i\in\mathcal{G}$ is represented by $K_i$ of those users.

The most natural strategy to employ in this situation would be to superpose the multi-user and the single-user strategies.
In particular, we divide the memory $M$ into $\gamma M$ and $(1-\gamma)M$, for some $\gamma\in[0,1]$.
We give the first part to $\mathcal{F}$ and the second part to $\mathcal{G}$, and apply their respective strategies on their part of the memory.
We believe this to be the best strategy, but proving its order-optimality requires developing new lower bounds that consider levels of both classes at the same time; this is part of our on-going work.

\bibliographystyle{IEEEtran}
\bibliography{caching}

\appendices
\section{Proofs for the multi-user case}

\subsection{Elaboration on the achievability (Theorem~\ref{thm:multi-user-achievability})}
\label{app:multi-user-achievability}

We here elaborate on the achievability of the multi-user setup, as the details will be important when proving Theorem~\ref{thm:multi-user-gap} in Appendix~\ref{app:multi-user-gap}.
The analysis is slightly different from the one in \cite{MLcodedcaching}.

As discussed in Section~\ref{sec:preliminaries}, the strategy involves finding a good partition $(H,I,J)$ of the set of levels.
Below we give such a partition.
\begin{definition}[$M$-feasible partition]
\label{def:m-feasible}
For any cache memory $M$, an $M$-feasible partition $(H,I,J)$ of the set of levels is a partition that satisfies:
\begin{IEEEeqnarray*}{lCl/rCcCl}
\forall h &\in& H, &
&& \tilde M &<& (1/K)\sqrt{N_h/U_h};\\
\forall i &\in& I, &
(1/K)\sqrt{N_i/U_i} &\le& \tilde M &\le& \left( 1+1/K \right)\sqrt{N_i/U_i};\\
\forall j &\in& J, &
\left( 1 + 1/K \right)\sqrt{N_j/U_j} &<& \tilde M,
\end{IEEEeqnarray*}
where $\tilde M = (M-T_J+V_I)/S_I$, and, for any subset $A$ of the levels:
\[
S_A = \sum_{i\in A} \sqrt{N_iU_i};
\quad
T_A = \sum_{i\in A} N_i;
\quad
V_A = \sum_{i\in A} N_i/K.
\]
\end{definition}
Such a partition always exists.
Furthermore, the set $I$ is never empty as long as $M\le\sum_iN_i$, \emph{i.e.}, as long as the caches do not have enough memory to store everything.

After choosing an $M$-feasible partition, we share the memory among the levels as follows:
\begin{IEEEeqnarray*}{lCl"rCl}
\forall h &\in& H, & \alpha_hM &=& 0;\\
\forall i &\in& I, & \alpha_iM &=& \sqrt{N_iU_i}\cdot\tilde M - N_i/K;\\
\forall j &\in& J, & \alpha_jM &=& N_j.
\end{IEEEeqnarray*}
It can be easily verified that $\alpha_i\in[0,1]$ for all $i$, and $\sum_i\alpha_i=1$.
Furthermore, we can also show that $\alpha_iM\le N_i$ for all $i$.

To properly analyze the achievable rate, we need to look more closely at the set $I$.
In the single-level scenario in \cite{maddah-ali2012,maddah-ali2013}, three regimes were identified, and they were analyzed separately.
These were: when $M<N/K$, when $M>cN$ for some constant $c\in(0,1)$, and the intermediate case.
We identify three similar regimes for each level in $i$.
Formally, let $\beta=1/80$, and define:
\begin{IEEEeqnarray*}{rCl}
I_0 &=& \left\{ i\in I : M < (2/K)\sqrt{N_i/U_i} \right\};\\
I_1 &=& \left\{ i\in I : M > (\beta+1/K)\sqrt{N_i/U_i} \right\};\\
I'  &=& I\setminus(I_0\cup I_1),
\IEEEyesnumber\label{eq:refined-partition}
\end{IEEEeqnarray*}
In other words, $I_0$ is the set of levels $i$ such that $\alpha_iM<N_i/K$, $I_1$ is such that $\alpha_iM>\beta N_i$ for all $i\in I_1$, and $I'$ is the rest.
For convenience, we call the resulting partition $(H,I_0,I',I_1,J)$ a \emph{refined $M$-feasible partition}.
Note that, because of regularity condition \eqref{eq:mu-reg2}, we have that $I_1$ will, at all times, contain at most one level.

We will now give upper bounds on the rates achieved individually for each level.
\begin{lemma}
\label{lemma:multi-user-achievability}
Given a refined $M$-feasible partition $(H,I_0,I',I_1,J)$, the individual rates of the levels are upper-bounded by:
\begin{IEEEeqnarray*}{lCl"rCl}
\forall h &\in& H, & R_h(M) &=& KU_h;\\
\forall i &\in& I_0\cup I', & R_i(M) &\le& \frac{2S_I\sqrt{N_iU_i}}{M-T_J+V_I};\\
\forall i &\in& I_1, &
R_i(M) &\le& \frac1\beta U_i\left( 1 - \frac{M-T_J}{N_i} \right)\\
&& & && {}+ \frac1\beta U_i\frac{S_{I_0}+S_{I'}}{\sqrt{N_iU_i}};\\
\forall j &\in& J, & R_j(M) &=& 0.
\end{IEEEeqnarray*}
\end{lemma}
For lack of space, we skip the proof of Lemma~\ref{lemma:multi-user-achievability}.

\subsection{Information-theoretic lower bounds}
\label{app:multi-user-converse}

The lower bounds we use are a sum of cut-set bounds.
Each such cut-set bound matches a popularity level:
it considers a certain number of caches that depends on the individual rate of the level, and then decodes files from just that level, resulting in an expression that matches the individual rate of the level.
To obtain this sum of different cut-set bounds, the \emph{sliding-window subset entropy inequality} \cite[Theorem~3]{Liu14} is used.
\begin{lemma}[Sliding-window subset entropy inequality {\cite[Theorem~3]{Liu14}}]
\label{lemma:sliding-window}
Given $K$ random variables $(Y_1,\ldots,Y_K)$, we have, for ever $s\in\{1,\ldots,K-1\}$:
\[
\frac1s\sum_{i=1}^K H\left(Y_i,\ldots,Y_{\langle i+s-1\rangle}\right)
\ge
\frac{1}{s+1}\sum_{i=1}^K
H\left( Y_i,\ldots,Y_{\langle i+s\rangle} \right),
\]
where we define $\langle i\rangle = i$ if $i\le K$ and $\langle i\rangle = i-K$ if $i>K$.
\end{lemma}

The resulting lower bounds are given in the following lemma.
\begin{lemma}
\label{lemma:multi-user-converse}
Consider the multi-level, multi-user caching setup.
Let $b\in\mathbb{N}^+$ and $t\in\{1,\ldots,K\}$, and let $s_i\in\{1,\ldots,\floor{K/2t}\}$ for every level $i$.
Then, for every memory $M$, the optimal rate can be bounded from below by:
\[
R^\ast(M)
\ge \sum_{i=1}^L \min\left\{ s_itU_i \,,\,\frac{N_i}{s_ib} \right\}
- \frac tb M.
\]
\end{lemma}

\begin{proof}
Without loss of generality, assume $s_1\le\cdots\le s_L$.
For every $i$ and $m$, define $\mathcal{Z}^m_i = (Z_i,\ldots,Z_{\langle i+m-1\rangle})$, where $Z_j$ denotes the $j$-th cache.
For every $i$, we consider the $t$ caches $\mathcal{Z}^t_i$ as well as $b$ broadcast messages $\mathcal{X}^b_i=(X_{i,1},\ldots,X_{i,b})$.
We have:
\[
bR + tM
\ge \frac1K \sum_{i=1}^K H\left( \mathcal{Z}^t_i,\mathcal{X}^b_i \right).
\]
By defining $Y_i=\left( \mathcal{Z}^t_i,\mathcal{X}^b_i \right)$ and rearranging the terms, we can use Lemma~\ref{lemma:sliding-window} to get:
\begin{IEEEeqnarray*}{rCll}
bR + tM
&\ge&
\frac1K \cdot \frac{1}{s_1} \sum_{i=1}^K &H\left(
\mathcal{Z}^t_i,
\mathcal{Z}^t_{\langle i+t \rangle},
\ldots,
Z^t_{\langle i+(s_1-1)t \rangle},
\right.\\
&& &\qquad
\left.
\mathcal{X}^b_i,
\mathcal{X}^b_{\langle i+t \rangle},
\ldots,
\mathcal{X}^b_{\langle i+(s_1-1)t \rangle}
\right).
\end{IEEEeqnarray*}
Thus, we have move from entropy terms with $t$ caches and $b$ broadcast messages to entropy terms with $s_1t$ caches and $s_1b$ broadcasts.
For simplicity, we write these entropy terms as $\hat H(s_1t,s_1b)$.
By applying a similar process for $i$ increasing from $1$ to $L$, we can get entropy terms of the form $\hat H(s_it,s_ib)$.
We use each such term to decode a certain number $p_i$ of files from level $i$
Since there are $s_it$ caches, $U_i$ level-$i$ users per cache, and $s_ib$ broadcasts, the total number of files that can be recovered is: $p_i=\min\{s_it\cdot U_i\cdot s_ib , N_i\}$.
Let $\mathcal{W}^i$ be this set of files.
Then:
\begin{IEEEeqnarray*}{rCl}
bR+tM
&\ge& \frac1K \sum_{i=1}^K \hat H(t,b)\\
&\overset{(a)}{\ge}& \frac1K \cdot \frac{1}{s_1} \sum_{i=1}^K \hat H(s_1t,s_1b)\\
&\overset{(b)}{\ge}& \frac1K \cdot \frac{1}{s_1} \sum_{i=1}^K \left(\hat H(s_1t,s_1b|\mathcal{W}^1) + p_1\right)\\
&\overset{(a)}{\ge}& \frac1K \cdot \frac{1}{s_2} \sum_{i=1}^K \hat H(s_2t,s_2b|\mathcal{W}^1) + \frac{p_1}{s_1}\\
&\overset{(b)}{\ge}& \frac1K \cdot \frac{1}{s_2} \sum_{i=1}^K \hat H(s_2t,s_2b|\mathcal{W}^1,\mathcal{W}^2) + \frac{p_1}{s_1} + \frac{p_2}{s_2}\\
&\ge& \cdots\\
&\ge& \sum_{i=1}^L \frac{p_i}{s_i}\\
&=& \sum_{i=1}^L \frac{1}{s_i} \min\left\{ s_i^2btU_i,N_i \right\};\\
\implies R
&\ge& \sum_{i=1}^L \min\left\{ s_itU_i \,,\, \frac{N_i}{s_ib} \right\} - \frac{t}{b} M,
\end{IEEEeqnarray*}
where inequalities marked with $(a)$ use the Lemma~\ref{lemma:sliding-window} and those marked with $(b)$ use Fano's inequality.
\end{proof}

\subsection{Proof of approximate optimality (Theorem~\ref{thm:multi-user-gap})}
\label{app:multi-user-gap}

For lack of space, we are unable to give the entire proof of Theorem~\ref{thm:multi-user-gap}.
It consists of several cases that are all analyzed in a similar manner, although with different values of certain parameters.
We will therefore give one such case in the hopes that it will be representative of the remaining cases.
This case is one where: $K\ge96$, $I_1=\emptyset$, and $J\not=\emptyset$.

Consider the lower bounds in Lemma~\ref{lemma:multi-user-converse}.
We will choose the parameters $t$, $b$, and $s_i$ such that the terms in the sum match the upper bounds in Lemma~\ref{lemma:multi-user-achievability}:
\begin{IEEEeqnarray*}{lCl"rCl}
&& & t &=& 1;\\
\forall h&\in&H, & s_h &=& \floor{K/8};\\
\forall i_0&\in&I_0, & s_{i_0} &=& \floor{(1/16)
  \frac{S_I\sqrt{N_{i_0}/U_{i_0}}}{M-T_J+V_I}};\\
\forall i&\in&I', & s_i &=& \floor{(1/8)
  \frac{S_I\sqrt{N_i/U_i}}{M-T_J+V_I}};\\
\forall j&\in&J, & s_j &=& 1;\\
&& & b &=& \floor{64\frac{(M-T_J+V_I)^2}{S_I^2}}.
\end{IEEEeqnarray*}

The first thing to do is to verify that these parameters satisfy their constraints.
The variables $t$ and $\{s_j\}_{j\in J}$ trivially do.

For $h\in H$, we have $s_h=\floor{K/8}\ge\floor{96/8}\ge1$ and $s_ht\le K$, thus $s_h$ satisfies all the constraints.

For $\{s_i\}_{i\in I'}$, we have, by \eqref{eq:refined-partition}:
\[
\frac{(1/8) S_I\sqrt{N_i/U_i}}{M-T_J+V_I}
\ge \frac{(1/8)\sqrt{N_i/U_i}}{(1/K+\beta)\sqrt{N_i/U_i}}
\ge 1,
\]
and hence $s_i\ge1$.
Furthermore:
\[
s_it
\le \frac{(1/8) S_I\sqrt{N_i/U_i}}{M-T_J+V_I}
\le \frac{(1/8)\sqrt{N_i/U_i}}{(2/K)\sqrt{N_i/U_i}}
\le K/2.
\]
Therefore, $s_i$ satisfies the constraints for $i\in I'$.
We can likewise show that, for $i_0\in I_0$, we have $s_{i_0}\ge1$ and $s_{i_0}t\le K/2$.

Finally, regarding the parameter $b$, we use the fact that there exists some level $j\in J$ to say:
\[
64\frac{(M-T_J+V_I)^2}{S_I^2}
\ge 64\left( \sqrt{N_j/U_j} \right)^2
= 64 N_j/U_j \ge 1,
\]
and hence $b\ge1$.

We must now evaluate the expression in Lemma~\ref{lemma:multi-user-converse} using these parameters.
For convenience, define $A_i=\min\{s_itU_i,N_i/s_ib\}$ for each level $i$.
We will lower-bound the value of $A_i$ for every $i$, which requires evaluating the following comparison:
\begin{equation}
\label{eq:mu-comparison}
bs_i^2t \overset{?}{\lessgtr} N_i/U_i.
\end{equation}

Consider $h\in H$.
Using Definition~\ref{def:m-feasible}, we can evaluate \eqref{eq:mu-comparison}:
\begin{IEEEeqnarray*}{rCl}
bs_h^2t
&\le& 64\frac{(M-T_J+V_I)^2}{S_I^2}\cdot(K/8)^2\cdot1\\
&\le& 64\cdot\frac{1}{K^2}\cdot\frac{N_h}{U_h}\cdot K^2/64\\
&\le& \frac{N_h}{U_h},
\end{IEEEeqnarray*}
which implies:
\begin{IEEEeqnarray*}{rCl}
A_h
&\ge& s_htU_h
= \floor{K/8}U_h
\ge (K/8-1)U_h\\
&\ge& KU_h \cdot (1/8-1/96)
= (11/96) KU_h.
\end{IEEEeqnarray*}

For $i_0\in I_0$, a similar process gives:
\[
A_{i_0}
\ge s_{i_0}tU_{i_0}
\ge (1/24) \frac{S_I\sqrt{N_{i_0}U_{i_0}}}{M-T_J+V_I},
\]
and, for $i\in I'$,
\[
A_i \ge s_itU_i
\ge (49/480) \frac{S_I\sqrt{N_iU_i}}{M-T_J+V_I}.
\]

Finally, for $j\in J$, we have:
\[
bs_j^2t = b \ge 32\frac{(M-T_J+V_I)^2}{S_I^2} \ge \frac{N_j}{U_j},
\]
again using Definition~\ref{def:m-feasible}.
As a result, we have $A_j=N_j/b$.

By combining all the $A_i$ values together, we get the following lower bound on the optimal rate (recall the values of $S_I$ and $T_J$ from Definition~\ref{def:m-feasible}):
\begin{IEEEeqnarray*}{rCl}
R^\ast(M)
&\ge& \sum_{h\in H} (11/96) KU_h
+ \sum_{i_0\in I_0} (1/24) \frac{S_I\sqrt{N_{i_0}U_{i_0}}}{M-T_J+V_I}\\
&& {} + \sum_{i\in I'} (49/480) \frac{S_I\sqrt{N_iU_i}}{M-T_J+V_I}
+ \sum_{j\in J} \frac{N_j}{b}
- \frac{M}{b}\\
&\ge& (11/96) \sum_{h\in H} KU_h
+ \frac{(1/24) S_I^2}{M-T_J+V_I}
- \frac{M-T_J}{b}\\
&\ge& (11/96) \sum_{h\in H} KU_h
+ (1/24) \frac{S_I^2}{M-T_J+V_I}\\
&& {} - \frac{M-T_J}{32(M-T_J+V_I)^2/S_I^2}\\
&\ge& (11/96) \sum_{h\in H} KU_h
+ (1/24) \frac{S_I^2}{M-T_J+V_I}\\
&& {} - \frac{1}{32}\cdot\frac{S_I^2}{M-T_J+V_I}\\
&=& (11/96) \sum_{h\in H} KU_h 
+ (1/96) \frac{S_I^2}{M-T_J+V_I}.
\IEEEyesnumber \label{eq:mu-lb}
\end{IEEEeqnarray*}

From Lemma~\ref{lemma:multi-user-achievability}, we can see that the achievable rate is bounded by:
\begin{equation}
\label{eq:mu-ub}
R(M) \le \sum_{h\in H} KU_h + \frac{2S_I^2}{M-T_J+V_I}.
\end{equation}
Combining \eqref{eq:mu-lb} with \eqref{eq:mu-ub}, we get that, in the case we are considering:
\[
\frac{R(M)}{R^\ast(M)} \le 192.
\]

The rest of the proof consists in carrying out a similar procedure for all other cases.
In the end, the worst-case gap between the achievable rate and the optimal rate is the maximum over the gaps found for each case.

\section{Proofs for the single-user case}

\subsection{Proof of the achievability (Theorem~\ref{thm:single-user-achievability})}

Recall how the memory is divided among the sets $H'$ and $I'$, defined in \eqref{eq:G-partition}: all of the available memory is given to $I'$, which is treated as one super-level.
As a result, all requests for files from $H'$ must be handled by a complete file transmission from the BS.
Since there are $\sum_{h\in H'}K_h$ users making such requests, the result is the same amount of transmissions.

For the set $I'$, now considered as one super-level, we use the single-level strategy from \cite{maddah-ali2013}.
Although only a subset of the caches is active in our setup, the same strategy still applies.
Indeed, the placement in \cite{maddah-ali2013} is a random sampling of the files in all the caches; we do the same placement in this case.
In the delivery phase, we now know the caches to which the users connected.
We perform a delivery as in \cite{maddah-ali2013}, assuming that only these caches were every present in the system.

As a result, the rate required for $I'$ can be directly derived from Lemma~\ref{lemma:single-level}, using $\sum_{i\in I'} K_i$ caches, $\sum_{i\in I'} N_i$ files, and $1$ user per cache.
In addition, we have, from \eqref{eq:G-partition}, that $M\ge N_i/K_i$ for all $i\in I'$.
This implies
\(
M \ge (\sum_{i\in I'}N_i)/(\sum_{i\in I'}K_i),
\)
and hence the rate for $I'$ is:
\[
R_{I'} = \frac{\sum_{i\in I'}N_i}{M}-1.
\]

It will be helpful for the later analysis to refine the partition $(H',I')$ as follows.
\begin{definition}
\label{def:G-partition}
Define the following partition $(G,H,I,J)$ of the set of levels:
\begin{IEEEeqnarray*}{rCl}
G &=& \left\{ g : M < N_g/K_g \text{ and } K_g\le5 \text{ and } M \le N_g/6 \right\};\\
H &=& \left\{ h : M < N_h/K_h \text{ and } K_h\ge6 \right\};\\
I &=& \left\{ i : N_i/K_i \le M \le N_i/6 \right\};\\
J &=& \left\{ j : M > N_j/6 \right\}.
\end{IEEEeqnarray*}
\end{definition}
Furthermore, we rewrite and bound the rate as follows:
\begin{equation}
\label{eq:single-user-achievable-rate}
R(M)
\le \sum_{g\in G} K_g
+ \sum_{h\in H} K_h
+ \frac{\sum_{i\in I}N_i}{M}
+ \left[ \frac{\sum_{j\in J}N_j}{M} - 1 \right]^+,
\end{equation}
where $[x]^+=\max\{x,0\}$.
We define $N_J=\sum_{j\in J}N_j$, and upper-bound the last term by:
\begin{equation}
\label{eq:su-ub-J}
\left[ \frac{N_J}{M} - 1 \right]^+
\le
\begin{cases}
N_J/M & \text{if $M<N_J/6$;}\\
6\left( 1 - M/N_J \right) & \text{if $N_J/6\le M < N_J$;}\\
0 & \text{if $M\ge N_J$.}
\end{cases}
\end{equation}

\subsection{Proof of approximate optimality (Theorem~\ref{thm:single-user-gap})}
\label{app:single-user-gap}

As previously mentioned, we use a cut-set bound to lower-bound the optimal rate.
The idea is to send a certain number $b$ of broadcast messages $X_1,\ldots,X_b$ that serve certain requests.
We choose these requests as follows.
For every level $i\in G\cup H\cup I$, consider a certain number $s_i\le K_i$ of caches.
These caches are distinct across levels.
For all the $b$ broadcasts, the users connected to these $s_i$ caches will altogether request $s_ib$ distinct files from level $i$ if there are that many; otherwise they request all $N_i$ files.
For the levels in the set $J$, we collectively consider some $s_J$ caches (distinct from the rest).
The users at these $s_J$ caches will use all $b$ broadcasts to decode as many files from the set $J$ as possible, up to $s_Jb$ files.
Let $n_J$ denote this number.

If we let $S=\sum_{i\not\in J}s_i + s_J$ be the total number of caches considered, then, by Fano's inequality:
\begin{IEEEeqnarray*}{rCl}
bR + SM
&\ge& H\left( Z_1,\ldots,Z_S,X_1,\ldots,X_b \right)\\
&\ge& \sum_{i\not\in J} \min\left\{ s_ib,N_i \right\} + n_J\\
R^\ast(M) &\ge& \sum_{i\not\in J} s_i\left( \min\left\{ 1,\frac{N_i}{s_ib} \right\} - \frac{M}{b} \right)\\
&& {} + s_J\left( \frac{n_J}{s_Jb} - \frac{M}{b} \right)\\
&=& \sum_{i\not\in J} v_i + v_J. \label{eq:single-user-lower-bounds}
\end{IEEEeqnarray*}

We will analyze each of the $v_i$ and $v_J$ terms separately.
We identify two cases for which the analysis is slightly different.

\subsubsection{Case $M<1/6$}

When $M$ is this small, we choose $b=1$ broadcast message.
Notice that, because of regularity condition \eqref{eq:su-reg}, we have $N_i/K_i\ge1>1/6>M$.

The achievable rate in this case can be upper-bounded by:
\begin{equation}
\label{eq:su-ub-0}
R(M) \le \sum_{i=1}^L K_i.
\end{equation}

Consider now any level $i$.
Let $s_i=K_i$.
Then,
\begin{IEEEeqnarray*}{rCl}
v_i
&=& s_i\left( \min\left\{ 1 , \frac{N_i}{s_ib} \right\} - \frac{M}{b} \right)\\
&=& K_i \left( \min\left\{ 1 , \frac{N_i}{K_i} \right\} - M \right)\\
&\ge& (5/6)K_i. \IEEEyesnumber \label{eq:su-lb-gh}
\end{IEEEeqnarray*}

We can combine \eqref{eq:su-ub-0} with \eqref{eq:su-lb-gh} and \eqref{eq:single-user-lower-bounds} to get:
\begin{equation}
\label{eq:su-gap-0}
\frac{R(M)}{R^\ast(M)} \le \frac65.
\end{equation}

\subsubsection{Case $M\ge1/6$}

We will now choose $b=\ceil{6M}\ge1$.

\paragraph{Bound for $g\in G$}
Consider $s_g=1$.
Then,
\begin{IEEEeqnarray*}{rCl}
v_g
&=& \min\left\{ 1 , \frac{N_i}{\ceil{6M}} \right\} - \frac{M}{\ceil{6M}}\\
&\ge& \frac12 - \frac{M}{6M}\\
&=& \frac13, \IEEEyesnumber \label{eq:su-lb-g}
\end{IEEEeqnarray*}
because $s_gb=\ceil{6M}\le2\cdot6M\le2N_g$.

\paragraph{Bound for $h\in H$}
Consider $s_h=\ceil{K_h/6}\ge1$.
Then,
\begin{IEEEeqnarray*}{rCl}
v_h
&\ge& \frac{K_h}{6} \left( \min\left\{ 1 , \frac{N_h}{4K_hM} \right\} - \frac{M}{6M} \right)
= \frac{1}{72}\cdot K_h, \IEEEyesnumber \label{eq:su-lb-h}
\end{IEEEeqnarray*}
because $s_hb=\ceil{K_h/6}\cdot\ceil{6M} \le 4K_hM \le 4N_h$.

\paragraph{Bound for $i\in I$}
Consider $s_i=\ceil{N_i/6M}$.
Then,
\begin{IEEEeqnarray*}{rCl}
v_i
&\ge& \frac{N_i}{6M} \left( \min\left\{ 1 , \frac{N_i}{4N_i} \right\} - \frac{M}{6M} \right)
= \frac{1}{72}\cdot\frac{N_i}{M}, \IEEEyesnumber \label{eq:su-lb-i}
\end{IEEEeqnarray*}
because $s_ib=\ceil{N_i/6M}\cdot\ceil{6M} \le 4N_i$.

\paragraph{Bound for $J$}
First, if $M\ge N_J$, then the set $J$ contributes nothing to the upper bound on the rate in \eqref{eq:single-user-achievable-rate}; see \eqref{eq:su-ub-J}.
Thus we can ignore it, \emph{i.e.}, say $v_J\ge0$.

So the interesting case is $M<N_J$.
Here, we must decode files from multiple levels collectively.
Consider $s_J=\ceil{N_J/6M}$, where $N_J=\sum_{j\in J}N_j$.
Notice that there are enough users and broadcasts to decode all files, because:
\[
s_Jb \ge \frac{N_J}{6M}\cdot6M = N_J.
\]
However, we must take care that no broadcast considers more than $K_j$ users at a time for any $j\in J$.
This can be ensured: since there are $b=\ceil{6M}$ broadcasts, and $b\ge N_j$ for all $j\in J$, then every broadcast need only consider at most one user per level.
Hence, all of the $N_J$ files can be decoded, and $n_J=N_J$.

If $M<N_J/6$, we have:
\begin{IEEEeqnarray*}{rCl}
v_J
&=& \frac{N_J - s_JM}{b}\\
&\ge& \frac{1}{12M}\left( N_J - \frac{N_J}{12M}\cdot M \right)\\
&\ge& \frac{144}{11}\cdot\frac{N_J}{M}.
\IEEEyesnumber \label{eq:su-lb-j1}
\end{IEEEeqnarray*}

If $N_J/6\le M<N_J$, then $s_J$ is actually equal to $1$, and:
\begin{IEEEeqnarray*}{rCl}
v_J
&=& \frac{N_J - M}{b}\\
&\ge& \frac{N_J - M}{12M}\\
&\ge& \frac{1}{12}\left( 1 - \frac{M}{N_J} \right).
\IEEEyesnumber\label{eq:su-lb-j2}
\end{IEEEeqnarray*}

\subsubsection{Multiplicative gap}

By combining \eqref{eq:su-lb-g}, \eqref{eq:su-lb-h}, \eqref{eq:su-lb-i}, \eqref{eq:su-lb-j1} and \eqref{eq:su-lb-j2} with \eqref{eq:single-user-achievable-rate} and \eqref{eq:su-ub-J}, and also taking into account \eqref{eq:su-gap-0}, we get:
\[
\frac{R(M)}{R^\ast(M)} \le 72,
\]
which concludes the proof of Theorem~\ref{thm:single-user-gap}.\qed

\end{document}